\documentclass[11pt]{amsart}

\usepackage[T1]{fontenc}
\usepackage{amsthm}
\usepackage{amsmath}
\usepackage{amsfonts}
\usepackage{tensor}
\usepackage{geometry}
\usepackage{mathrsfs}
\usepackage{setspace}
\usepackage{fancyhdr}
\usepackage{enumerate}

\newcommand{\N}{\mathbb N}

\newcommand{\C}{\mathbb C}
\newcommand{\R}{\mathbb R}

\newcommand{\la}{\langle}
\newcommand{\ra}{\rangle}

\newcommand{\hi}{\mathcal H}

\newcommand{\ki}{\mathcal K}

\makeatletter
\def\blfootnote{\xdef\@thefnmark{}\@footnotetext}
\makeatother

\newtheorem{thm}{Theorem}[section]
\newtheorem{cor}[thm]{Corollary}
\newtheorem{prop}[thm]{Proposition}
\newtheorem{rmk}[thm]{Remark}
\newtheorem{exm}[thm]{Example}
\newtheorem{deff}[thm]{Definition}

\begin{document}
\title{Generators of Quantum Markov Semigroups}

\author{George Androulakis}
\address{Department of Mathematics, University of South Carolina, 
Columbia, SC 29208}
\email{giorgis@math.sc.edu}

\author{Matthew Ziemke}
\address{Department of Mathematics, University of South Carolina, 
Columbia, SC 29208}
\email{ziemke@email.sc.edu}

\blfootnote{The article is part of the 
second author's Ph.D. thesis which is prepared at the University of South Carolina under the supervision of the first author.}

%
%
\begin{spacing}{1.4}
\maketitle
\begin{abstract}
Quantum Markov Semigroups (QMSs) originally arose in the study of the evolutions of irreversible open quantum systems.  Mathematically, they are a generalization of classical Markov semigroups where the underlying function space is replaced by a non-commutative operator algebra.  In the case when the QMS is uniformly continuous, theorems due to Lindblad \cite{lindblad}, Stinespring \cite{stinespring}, and Kraus \cite{kraus} imply that the generator of the semigroup has the form
$$L(A)=\sum_{n=1}^{\infty}V_n^*AV_n +GA+AG^*$$
where $V_n$ and G are elements of the underlying operator algebra.  In the present paper we investigate the form of the generators of QMSs which are not necessarily uniformly continuous and act on the bounded operators of a Hilbert space.  We prove that the generators of such semigroups have forms that reflect the results of Lindblad and Stinespring.  We also make some progress towards forms reflecting Kraus' result.  Lastly we look at several examples to clarify our findings and verify that some of the unbounded operators we are using have dense domains. 
\end{abstract}
\section{Motivation and Overview of our Results}
In this section we motivate and overview our results while precise definitions appear in section \ref{sn2}.  In the early seventies, R.S. Ingarden and A. Kossakowski (see \cite{ingarden-kossakowski} and \cite{kossakowski}) postulated that the time evolution of a statistically open system, in the Schrodinger picture, be given by a one-parameter semigroup of linear operators acting on the trace-class operators of a separable Hilbert space $\hi$ satisfying certain conditions.  In the Heisenberg picture the situation translates to a one-parameter semigroup $(T_t)_{t \geq 0}$ acting on $\mathcal{B}(\hi)$ (the bounded operators on a Hilbert space $\hi$) where each $T_t$ is positive and $\sigma$-weakly continuous, satisfying $T_t(1)=1$ for all $t \geq 0$, and where the map $t \mapsto T_tA$ is $\sigma$-weakly continuous for each $A \in \mathcal{B}(\hi)$.  
\\In 1976, G. Lindblad \cite{lindblad} added to the formulation the condition that each $T_t$ be completely positive rather than simply positive, a condition which he justified physically.  Results of Stinespring \cite[Theorem 4]{stinespring} and Arveson \cite[Proposition 1.2.2]{arveson2} further justify this condition by proving that if an operator has a commutative domain or target space then positivity and complete positivity are equivalent.  Further, under the assumption that the map $t \mapsto T_t$ is uniformly continuous, the semigroup is called a uniformly continuous QMS, the generator L of the semigroup is bounded, and Lindblad was able to write L in the form $L(A)=\phi(A)+G^*A+AG$ where $\phi$ is completely positive and $G \in \mathcal{B}(\hi)$.
Using an earlier theorem of Stinespring \cite{stinespring} we can then write $\phi$ in the form $\phi(A)=V^* \phi (A)V$ where $V: \hi \rightarrow \ki$ for some Hilbert space $\ki$ and $\pi: \mathcal{B}(\hi) \rightarrow \mathcal{B}(\ki)$ is a normal representation.  Further, a theorem due to Kraus \cite{kraus} lets us write $\pi$ in the form $\pi (A)=\sum_{n=1}^{\infty}W_n^*AW_n$ where $W_n: \ki \rightarrow \hi$ is a bounded linear operator.  When we combine Stinespring's and Kraus' results we are then able to write $\phi$ in the form $\phi(A)=\sum_{n=1}^{\infty}V_n^*AV_n$ where $V_n \in \mathcal{B}(\hi)$.  Lindblad's original result was for QMSs on a hyperfinite factor $\mathfrak{A}$ of $\mathcal{B}(\hi)$ (which includes the case $\mathfrak{A}=\mathcal{B}(\hi)$, see \cite{topping}).  A similar result to Lindblad's was given in that same year by Gorini, Kossakowski, and Sudarshan in \cite{gks} for QMSs on finite dimensional Hilbert spaces and three years later Christensen and Evans proved it for uniformly continuous QMSs on arbitrary von Neumann algebras in \cite{ce}.  A nice exposition of these results is written by Fagnola~\cite{fagnola1}.  Another name for QMSs that appears in the literature is $CP_0$-semigroups \cite{arveson3}.  An important subclass of QMSs that has also attracted a lot of attention is the class of $E_0$-semigroups which was introduced by Powers \cite{powers2}.
\\In this paper we prove analogous results to Lindblad and Stinespring and make some progress towards Kraus for the generator of a QMS acting on $\mathcal{B}(\hi)$ when we no longer assume that the semigroup is uniformly continuous.  In this case, the generator L is no longer bounded and so inevitably, much discussion on domains of operators and the density of such domains is required.  Because of such difficulties we introduce the notion of U-completely positive maps (for a linear subspace U of $\hi$) which is analogous to completely positive maps but is better suited for unbounded operators (see Definition \ref{3.3}).  We are then able to show (see Theorem \ref{3.6}) that if L denotes the generator of a QMS on $\mathcal{B}(\hi)$ then there exists a subspace $W$ of $\hi$, a linear operator $K:W \rightarrow \hi$, and a W-completely positive map $\phi:D(L) \rightarrow S(W)$ (where $D(L)$ denotes the domain of L and $S(W)$ denotes the set of sesquilinear forms on $W \times W$) such that
$$\la u, L(A)v \ra = \phi (A)(u,v)+ \la Ku,Av \ra +\la u, AKv \ra$$ for all $A \in D(L)$ and all $u,v \in W$.
Unfortunately this result does not tell us much about the subspace W or the operator K.  On the other hand, if we restrict ourselves to the domain algebra $\mathcal{A}$ of L, which is the largest $^*$-subalgebra of the domain of L and was studied by Arveson \cite{arveson}, then we are able to find (see Theorem \ref{3.9}) an explicit subspace U of $\hi$ and a linear operator $G:U \rightarrow \hi$ having an explicit formula and a U-completely positive map $\phi : \mathcal{A} \rightarrow S(U)$ such that 
$$\la u, L(A)v \ra = \phi(A)(u,v) + \la u, GA v \ra + \la GA^*u,v \ra$$
for all $A \in \mathcal{A}$ and for all $u,v \in U$ where $\phi : \mathcal{A} \rightarrow S(U)$ is U-completely positive.
\\With regard to Stinespring, we are able to show (see Theorem \ref{3.4}) that there exists a Hilbert space $\ki$, a linear map $V:\hi \rightarrow \ki$, and a unital $^*$-representation $\pi: \mathcal{A} \rightarrow \mathcal{B}(\hi)$ so that $\phi(A)(u,w)= \la Vu, \pi(A)Vw \ra$ for all $u,w \in U$.  Theorems \ref{3.9} and \ref{3.4} are summarized in Corollary \ref{3.7} which is the main result of our paper.  In Section \ref{attemptkraus} we give partial results similar to the one given by Kraus but fall slightly short and discuss a possible way forward (see Proposition \ref{5.4} and the discussion that follows it).  Finally in Section \ref{sn6} we look at three examples to verify the form of their generators and to discuss their corresponding subspace U mentioned above.

\section{Mathematical Background}\label{sn2}
In this section we provide the necessary definitions and mathematical background that is needed for the rest of the paper.  Throughout the paper, $\hi$ will denote a Hilbert space.  To avoid confusion we want to mention from the start that all of our inner products are linear in the second coordinate and conjugate linear in the first.  Also, for $x,y \in \hi$, we define the rank one operator $|x \ra \la y|: \hi \rightarrow \hi$ by $|x \ra \la y|(h)= \la y , h \ra x$.  We will extensively use the $\sigma$ -weak topology so it is worth recalling:  On a general von Neumann algebra, the $\sigma$-weak topology is the $w^*$ topology given by its predual (every von Neumann algebra has a predual).  If the von Neumann algebra under consideration is $\mathcal{B}(\hi)$ then the predual is given by the space of all trace class operators on $\hi$ which we'll denote by $L_1(\hi)$.  For a detailed description of the duality between $\mathcal{B}(\hi)$ and $L_1(\hi)$ we refer the reader to \cite[Theorem 3.4.13]{pedersen}.

\begin{deff}\label{1.0}
Let $\mathfrak{A}$ be a von Neumann algebra and let $M_n$ be the set of all $n \times n$ matrices with complex coefficients.  Then the algebraic tensor product $\mathfrak{A} \otimes M_n$ can be represented as the $^*$-algebra of $n \times n$ matrices with entries in $\mathfrak{A}$.  Every element $A \in \mathfrak{A} \otimes M_n$ can be written in the form
$$A= \sum_{i,j=1}^n A_{ij} \otimes E_{ij}$$
where $E_{ij}$ is the $n\times n$ matrix with 1 in the (i,j)th position and zero everywhere else.  If $\mathfrak{B}$ is also a von Neumann algebra and $T:\mathfrak{A} \rightarrow \mathfrak{B}$ is a linear operator then we define the linear map $T^{(n)}: \mathfrak{A} \otimes M_n \rightarrow \mathfrak{B} \otimes M_n$ by
$$T^{(n)} \left( \sum_{i,j=1}^n A_{ij} \otimes E_{ij} \right)= \sum_{i,j=1}^n T(A_{ij}) \otimes E_{ij}$$
We say a map $T:\mathfrak{A} \rightarrow \mathfrak{B}$ is  $\bf{positive}$ if it maps positive elements to positive elements.  It is called $\bf{completely}$ $\bf{positive}$ if $T^{(n)}$ is positive for all $n \in \N$.  In the case that $\mathfrak{B}$ acts on a Hilbert space $\hi$ it can be proven that T is completely positive if
$$\sum_{i,j=1}^n \la h_i , T(A_i^*A_j)h_j \ra \geq 0$$
for all $n \in \N$, $A_1, \dots , A_n \in \mathfrak{A}$, and $h_1, \dots , h_n \in \hi$ \cite[Proposition 2.9]{fagnola1}.
\end{deff}

\begin{deff}\label{1.1}
Let $\mathfrak{A}$ be a von Neumann algebra.  A $\bf{Quantum}$ $\bf{Dynamical}$ $\bf{Semigroup}$ $\bf{(QDS)}$ is a one-parameter family $(T_t)_{t \geq 0}$ of $\sigma$-weakly continuous, completely positive, linear operators on $\mathfrak{A}$ such that
\\(i)  $T_0=1$
\\(ii)  $T_{t+s}=T_tT_s$
\\(iii)  for a fixed $A \in \mathfrak{A}$, the map $t \mapsto T_t(A)$ is $\sigma$-weakly continuous.
\\Further, if $T_t(1)=1$ for all $t \geq 0$ then we say the quantum dynamical semigroup is $\bf{Markovian}$ or we simply refer to it as a $\bf{Quantum}$ $\bf{Markov}$ $\bf{Semigroup}$ $\bf{(QMS)}$.  If the map  $t \mapsto T_t$ is norm continuous then we say the semigroup is $\bf{uniformly}$ $\bf{continuous}$.
\end{deff}
\noindent$\bf{Note}$:  If $(T_t)_{t \geq 0}$ is a Quantum Markov Semigroup then $\| T_t \|=1$ for all $t \geq 0$.  This is due to \cite[Corollary 1]{dye-russo}.

\begin{deff}\label{1.2}
Given a QDS $(T_t)_{t \geq 0}$, we say that an element $A \in \mathfrak{A}$ belongs to the domain of the infinitesimal generator $L$ of $(T_t))_{t \geq 0}$, denoted by $D(L)$, if
$$ \lim_{t \rightarrow 0} \frac{1}{t} ( T_tA - A)$$
converges in the $\sigma$-weak topology and, in this case, define the $\bf{infinitesimal}$ $\bf{generator}$ to be the generally unbounded operator $L$ such that
$$ L(A)=\sigma \text{-weak-} \lim_{t \rightarrow 0} \frac{1}{t} ( T_tA - A) \quad , \quad A \in D(L) .$$
If $(T_t)_{t \geq 0}$ is uniformly continuous then the generator $L$ is bounded and given by
$$L=\lim_{t \rightarrow 0} \frac{1}{t}(T_t -1)$$
where the limit is taken in the norm topology.
\end{deff}
It has been proven (see \cite[Proposition 3.1.6]{br}) that the domain of the generator L of a QDS is $\sigma$-weakly dense.
However, if the QDS is not uniformly continuous then generator L does not have full domain.  Indeed, it is known (see \cite[Proposition 3.1.6]{br}) that L is $\sigma$-weakly closed so if L has full domain then it would be bounded.  In this case the QDS is then uniformly continuous (see \cite{hille-phillips}).

\section{Generators of Uniformly Continuous Quantum Markov Semigroups on $\mathcal{B}(\hi)$}
In this section we recall some results for the form of the generator of a uniformly continuous QMS (which motivate our work on the consequent sections) and we improve existing results.  As a motivation for Lindblad's result we start by describing a simple example of a QDS and its generator which comes from \cite[Example 3.1]{fagnola1}.

\begin{exm}\label{1.3}
Let $(U_t)_{t \geq 0}$ be a strongly continuous semigroup on a Hilbert space $\hi$.  Then, define $T_t: \mathcal{B}(\hi) \rightarrow \mathcal{B}(\hi) $, for all $t \geq 0$, by
$$T_t(A)=U_tAU_t^*$$
Then $(T_t)_{t \geq 0}$ is a quantum dynamical semigroup.  Further, if $G$ is the generator of $(U_t)_{t \geq 0}$ and $G$ is bounded then the generator, L, of $(T_t)_{t \geq 0}$ is given by
$$L(A)=GA+AG^* .$$
\end{exm}

This form should be compared with \eqref{lindblad} of Theorem \ref{2.2} (Lindblad's result).  In Theorem \ref{2.2} we give a proof of Lindblad for the case of QMSs defined on $\mathcal{B}(\hi)$ which allows for a great deal of possibilities for the operator $G$ in the formula of $L(A)$ which appears in the Abstract.  The following result has been proven in \cite[Lemma 3.13]{fagnola1} for the case of uniformly continuous QDS.  Here we remove the uniform continuity assumption.

\begin{prop}\label{2.1}
If $L$ is the generator of a QDS on $\mathcal{B}(\hi)$ and $\mathfrak{A}$ is a $^*$-subalgebra of $\mathcal{B}(\hi)$ such that $\mathfrak{A} \subseteq D(L)$ then, for all $A_1, \dots , A_n \in \mathfrak{A}$ and $u_1, \dots , u_n \in \hi$ such that $\sum_{k=1}^n A_ku_k = 0$, we have that 

$$\sum_{i,j=1}^n \la u_i , L(A_i^*A_j) u_j \ra \geq 0 .$$
\end{prop}
\begin{proof}
We start with a $\bf{claim}$:  If $(T_t)_{t \geq 0}$ is a $\sigma$-weakly continuous semigroup of positive operators and $L$ is the generator then, for any $A \in \mathfrak{A}$ and $u \in \hi$ such that $Au=0$ we have that $ \la u , L(A^*A)u \ra \geq 0$.
\\Indeed, for $u \in \hi$ define $T: \hi \rightarrow \hi$ by $Th= \la u , h \ra u=|u\ra \la u|(h)$.  Clearly $T$ is rank one and hence $T$ is a trace class operator on $\hi$.  Further, if $\varphi_T$ is the image of $T$ in $\mathfrak{A}^*$ under the trace duality then
$$\varphi_T(B)=tr(BT)= \la u , Bu \ra$$
for all $B \in \mathfrak{A}$.  Then, for $A \in \mathfrak{A}$ such that $Au=0$ we have
$$\la u , L(A^*A)u \ra=\varphi_T (L(A^*A)) = \lim_{\epsilon \rightarrow 0} \frac{1}{\epsilon} \varphi_T(T_{\epsilon}(A^*A)-A^*A) .$$
Further,
\begin{align*}
\lim_{\epsilon \rightarrow 0} \frac{1}{\epsilon} \varphi_T(T_{\epsilon}(A^*A)-A^*A) & = \lim_{\epsilon \rightarrow 0} \frac{1}{\epsilon} \left( \la u , T_{\epsilon}(A^*A)u \ra - \la u , A^*A u \ra \right)
\\ & = \lim_{\epsilon \rightarrow 0} \frac{1}{\epsilon} \la u , T_{\epsilon}(A^*A)u \ra && \text{since } Au=0
\\ & \geq 0 && \text{ since } T_{\epsilon} \geq 0
\end{align*}
which completes the proof of the claim.
\\Now, suppose  $A_1, \dots , A_n \in \mathfrak{A}$ and $u_1, \dots , u_n \in \hi$ such that $\sum_{k=1}^n A_ku_k = 0$.  Since $T_t$ is completely positive, $T_t^{(n)}$ is positive.  So, $(T_t^{(n)})_{t \geq 0}$ is a $\sigma$-weakly continuous semigroup of positive operators with generator $L^{(n)}$.  Let $A_0=\sum_{k=1}^n A_k \otimes E_{1,k}$ and let $u_0=(u_1, \dots u_n)^T$ (where T stands for transpose).  Then, by the above claim,
$$ 0 \leq \la u_0 , L^{(n)} (A_0^*A_0) u_0 \ra = \sum_{j,k=1}^n \la u_j , L(A_j^*A_k) u_k \ra$$
which completes the proof.  
\end{proof}

We will now proceed to look at a proof of Lindblad's Theorem for uniformly continuous QMSs on $\mathcal{B}( \hi)$.  Lindblad's original proof was for any hyperfinite factor in $\mathcal{B}(\hi)$.  Our proof was motivated by a proof given in \cite[Theorem 3.14]{fagnola1}, but as stated earlier, gives us more options in defining the operator $G$ in the formula of $L(A)$ which appears in equation \eqref{lindblad} below.  We make use of the greater flexibility of the form of G in Theorem \ref{3.9}.

\begin{thm}[Lindblad]\label{2.2}
 Let $L$ be the generator of a uniformly continuous QMS on $\mathcal{B}(\hi)$.  Let $T$ be any positive finite rank operator on $\hi$.  Then there exists $h \in \hi$ such that if the operator $G$ is defined on $\hi$ by
$$G(x)=L(|x \ra \la Th |)h-\frac{1}{2} \la h, L(T)h \ra x$$
then there exists a completely positive map $\phi : \mathcal{B}(\hi) \rightarrow \mathcal{B}( \hi)$ such that 
\begin{equation}\label{lindblad}
  L(A)= \phi (A) + GA+AG^*
\end{equation}
for all $A \in \mathcal{B}(\hi)$.
\end{thm}

\begin{proof}
By the spectral theorem for compact self-adjoint operators we have that for any positive finite rank operator $T$ there exist finitely many orthonormal vectors $(k_s')_{s=1}^m$, and positive numbers $(t_s')_{s=1}^m$ such that $T= \sum_{s=1}^mt_s' |k_s' \ra \la k_s'|$.  If we define $t= \sum_{s=1}^m t_s'$, $t_s =t_s'/t$, and $k_s=\sqrt{t}k_s'$ then we can rewrite $T$ as $T=\sum_{s=1}^mt_s |k_s \ra \la k_s|$ where $t_s \geq 0$, $\sum_{s=1}^m t_s=1$, and $\la k_{s_1}, k_{s_2} \ra =0$ if $s_1 \neq s_2$.  Let $h=\sum_{s=1}^m k_s\|k_s \|^{-2} \in \hi$.  Then $\la h, k_s \ra =1$ for all $s=1, \dots , m$.
\\ {\bf Claim}:  For $s=1, \dots, m$ if we define the operator $G_s: \hi \rightarrow \hi$ by
\begin{equation}\label{2.2.1}
G_s(x)=L(|x \ra \la k_s|)h-\frac{1}{2} \la h , L (|k_s \ra \la k_s |)h \ra x
\end{equation}
and $\phi_s:  \mathcal{B} (\hi) \rightarrow \mathcal{B} (\hi)$ by
$$\phi_s(A)=L(A) - G_sA -AG_s^*$$
then $\phi_s$ is completely positive.
\\Once the claim is proved, then the map $\phi: \mathcal{B}(\hi) \rightarrow \mathcal{B}(\hi)$ defined by $\phi = \sum_{s=1}^mt_s \phi_s$ is completely positive since the coefficients $t_s$ are non-negative.  Since $\sum_{s=1}^mt_s=1$, we have that
\begin{equation}\label{2.2.2}
\phi (A)=L(A)- \left( \sum_{s=1}^mt_sG_s \right)A - A  \left( \sum_{s=1}^mt_sG_s^* \right)
\end{equation}
and
$$ \sum_{s=1}^mt_sG_s^* =\left( \sum_{s=1}^mt_sG_s \right)^* .$$
Hence, if we set $G= \sum_{s=1}^mt_sG_s$, \eqref{2.2.2} gives \eqref{lindblad}.
Note that by multiplying \eqref{2.2.1} by $t_s$ and summing up we obtain
\begin{align*}
G(x)=\left( \sum_{s=1}^mt_sG_s \right)(x) & =L \left(|x \ra \la \sum_{s=1}^mt_sk_s| \right)h - \frac{1}{2} \la h , L \left( \sum_{s=1}^m t_s |k_s \ra \la k_s| \right) h \ra x
\\ & = L \left(| x \ra \la Th | \right)h - \frac{1}{2} \la h, L(T)h \ra x .
\end{align*}
Thus it only remains to prove the claim.  Fix $s \in \{ 1, \dots , m \}$.  We vary the technique of \cite[Theorem 3.14]{fagnola1} as follows.  Let $A_1, \dots , A_n \in \mathcal{B}(\hi)$ and $h_1, \dots , h_n \in \hi$.  Let $v=- \sum_{i=1}^nA_ih_i$, $A_{n+1}= |v \ra \la k_s |$ and $h_{n+1}=h$.  Then, since $\la h, k_s \ra=1$,
$$\sum_{i=1}^{n+1}A_ih_i=\sum_{i=1}^nA_ih_i+A_{n+1}h_{n+1}=-v+|v \ra \la k_s|(h)=-v+v=0 .$$
Since $L$ is the generator of a uniformly continuous QMS, by Proposition \ref{2.1},
\begin{align*}
0 & \leq \sum_{i,j=1}^{n+1} \la h_i, L(A_i^*A_j)h_j \ra
\\ & =\sum_{i,j=1}^n \la h_i, L(A_i^*A_j)h_j \ra +\sum_{i=1}^n \la h_i, L(A_i^*A_{n+1})h_{n+1} \ra +\sum_{j=1}^n \la h_{n+1}, L(A_{n+1}^*A_j)h_j \ra 
\\ & + \la h_{n+1}, L (A_{n+1}^*A_{n+1})h_{n+1} \ra .
\end{align*}
Hence,
\begin{align*}
0 & \leq \sum_{i,j=1}^n \la h_i, L(A_i^*A_j)h_j \ra + \sum_{i=1}^n \la h_i, L(|A_i^*(v) \ra \la k_s|)h \ra + \sum_{j=1}^n \la h, L(|k \ra \la A_j^*(v)|)h_j \ra 
\\ & + \| v \|^2 \la h , L(|k_s \ra \la k_s| ) h \ra
\\ & =\sum_{i,j=1}^n \la h_i, L(A_i^*A_j)h_j \ra - \sum_{i,j=1}^n \la h_i, L(|A_i^*A_jh_j \ra \la k_s|)h \ra - \sum_{i,j=1}^n \la h, L(|k_s \ra \la A_j^*A_ih_i|)h_j \ra 
\\ & + \sum_{i,j=1}^n \la A_ih_i,A_jh_j \ra \la h, L(|k_s \ra \la k_s|)h \ra .
\end{align*}
 If we break up the last term into two equal pieces and subtract each from the second and third term of the last expression, then we obtain
\begin{align*}   
0 & \leq \sum_{i,j=1}^n \left[ \la h_i, L(A_i^*A_j)h_j \ra - \left( \la h_i, L(|A_i^*A_jh_j \ra \la k_s|)h \ra - \frac{1}{2}\la h_i,A_i^*A_jh_j \ra \la h, L(|k_s \ra \la k_s|)h \ra \right) \right.
\\ & - \left. \left( \la L(|A_j^*A_ih_i \ra \la k_s| ) h,h_j \ra - \frac{1}{2} \la A_j^*A_ih_i, h_j \ra \la h , L(| k_s \ra \la k_s|) h \ra \right) \right] .
\end{align*}
Define an operator $G_s: \hi \rightarrow \hi$ by $G_s(x)=L(|x \ra \la k_s| ) h - \frac{1}{2} \la h , L (|k_s \ra \la k_s|) (h) \ra x$ to continue
\begin{align*}
0 & \leq \sum_{i,j=1}^n \left( \la h_i, L(A_i^*A_j)h_j \ra - \la h_i,G_sA_i^*A_jh_j \ra - G_sA_j^*A_ih_i,h_j \ra \right)
\\ & = \sum_{i,j=1}^n \la h_i, \left(L(A_i^*A_j)-G_sA_i^*A_j-A_i^*A_jG_s^* \right)h_j \ra = \sum_{i,j=1}^n \la h_i , \phi_s(A_i^*A_j)h_j \ra
\end{align*}
which finishes the proof of the claim and the theorem.  

\end{proof}

\begin{deff}\label{3.8}
Let T be a positive finite rank operator in $\mathcal{B}(\hi)$.  Then we will call the vector $h \in \hi$, as defined in Theorem \ref{2.2}, an $\bf{ associate}$ vector for $T$.
\end{deff}

We have casually mentioned the results of Stinespring \cite{stinespring} and Kraus \cite{kraus} earlier.  Since we will attempt to generalize both, we feel it is necessary to give complete statements of them.

\begin{thm}[Stinespring]\label{2.3}
Let $\mathfrak{B}$ be a $C^*$-subalgebra of the algebra of all bounded operators on a Hilbert space $\hi$ and let $\mathfrak{A}$ be a $C^*$-algebra with unit.  A linear map $T:\mathfrak{A} \rightarrow \mathfrak{B}$ is completely positive if and only if it has the form
\begin{equation}\label{stinespring}
T(A)=V^* \pi (A) V
\end{equation}
where $(\pi, \ki )$ is a unital $^*$-representation of $\mathfrak{A}$ on some Hilbert space $\ki$, and V is a bounded operator from $\hi$ to $\ki$.
\end{thm}

\begin{thm}[Kraus]\label{2.4}
Let $\mathfrak{A}$ be a von Neumann algebra of operators on a Hilbert space $\hi$ and let $\ki$ be another Hilbert space.  A linear map $T: \mathfrak{A} \rightarrow \mathcal{B}(\ki)$ is normal and completely positive if and only if it can be represented in the form
\begin{equation}\label{kraus}
T(A)= \sum_{j=1}^{\infty}V_j^*AV_j
\end{equation}
where $(V_j)_{j=1}^{\infty}$ is a sequence of bounded operators from $\ki$ to $\hi$ such that the series $\sum_{j=1}^{\infty}V_j^*AV_j$ converge strongly.
\end{thm}

\section{Generators of General Quantum Markov Semigroups on $\mathcal{B}(\hi)$}
In this section we prove analogous expressions of \eqref{lindblad} and \eqref{stinespring} for the generator of a general QMS on $\mathcal{B}(\hi)$.  The main result of the section as well as the main result of the paper is Corollary \ref{3.7}.  Heading in this direction, we start with the following:

\begin{thm}\label{3.1}
Let $L$ be the generator of a QMS on $\mathcal{B}(\hi)$.  Then there exists a family $(L_{\epsilon})_{\epsilon >0}$ of generators of uniformly continuous QMSs on $\mathcal{B}(\hi)$ such that
$$L(A)= \lim_{\epsilon \rightarrow 0} L_{\epsilon}(A)$$
for all $A \in D(L)$, where the limit is taken in the $\sigma$-weak topology.  Thus, by Theorem \ref{2.2}, there exists a family $(\phi_{\epsilon})_{\epsilon >0}$ of normal completely positive operators on $\mathcal{B}(\hi)$ and a family $(G_{\epsilon})_{\epsilon >0}$ of bounded operators on $\hi$ such that
$$L_{\epsilon}(A)= \phi_{\epsilon}(A)+G_{\epsilon}A+AG_{\epsilon}^*$$
for all $A \in \mathcal{B}(\hi)$.
\end{thm}
\begin{proof}
Let $L$ be the generator for a Quantum Markov Semigroup $(U_t)_{t \geq 0}$.  Let $L_{\epsilon}=L(1- \epsilon L)^{-1}$.  Then, for $\epsilon >0$, $L_{\epsilon}$ is bounded and $\sigma$-weakly continuous, since by Proposition 3.1.4 and Proposition~3.1.6 of \cite{br}, $(1- \epsilon L)^{-1}$ is bounded and $\sigma$-weakly continuous and
\begin{equation}\label{3.1.1}
L(1-\epsilon L)^{-1}=-\frac{1}{\epsilon} \left(1-(1- \epsilon L)^{-1} \right) .
\end{equation}
Define $U_{t, \epsilon}:\mathcal{B} (\hi) \rightarrow \mathcal{B} ( \hi)$ by $U_{t, \epsilon}= \exp{(tL_{\epsilon})}$.  Then we know $(U_{t, \epsilon})_{t \geq 0}$ is a uniformly continuous semigroup.  Further, we claim that $(U_{t, \epsilon})_{t \geq 0}$ is contractive.  Indeed, by \cite[Theorem 3.1.10]{br} we have that $\| (1- \epsilon L)^{-1} \| \leq 1$ for all $\epsilon >0$, so
\begin{align*}
\| U_{t, \epsilon} \| & = \|e^{tL_{\epsilon}} \| \leq e^{-t/ \epsilon} \sum_{n=0}^{\infty} \frac{ (t/ \epsilon)^n}{n!} \| (1- \epsilon L)^{-n} \| && \text{by \eqref{3.1.1}}
\\ & \leq e^{-t/ \epsilon} \sum_{n=0}^{\infty} \frac{ (t/ \epsilon)^n}{n!} = 1
\end{align*}
and so $(U_{t, \epsilon})_{t \geq 0}$ is contractive.  Further, since $L_{\epsilon}$ is $\sigma$-weakly continuous we have, by \cite[Proposition 3.9]{fagnola1}, that $U_{t, \epsilon}$ is $\sigma$-weakly continuous.  Also, since $(U_t)_{t \geq 0}$ is Markovian, $1 \in D(L)$ and $L(1)=0$ so
$$L_{\epsilon}(1)=L(1- \epsilon L)^{-1}(1)=(1- \epsilon L)^{-1}L(1)=0 .$$
Hence
$$U_{t, \epsilon}(1) = 1 + \sum_{n=1}^{\infty} \frac{t^n}{n!}L_{\epsilon}^n(1) =1 .$$ 
So, $\| U_{t, \epsilon} \|=1$ and the norm is attained at 1 so, by \cite[Corollary 1]{dye-russo}, $U_{t, \epsilon}$ is positive.  Now $(U_t^{(n)})_{t \geq 0}$ is also a Quantum Markov Semigroup with generator $L^{(n)}$ so, following the above with $U_t^{(n)}$ in place of $U_t$ and $L^{(n)}$ in place of $L$ we get that $\exp{(tL^{(n)}(1-\epsilon L^{(n)})^{-1})} \geq 0$ for all $n \in \N$.  We now claim that $L^{(n)}(1-\epsilon L^{(n)})^{-1}=(L(1-\epsilon L)^{-1})^{(n)}$ which will prove that $U_{t,\epsilon}$ is completely positive, since $(L(1-\epsilon L)^{-1})^{(n)}$ is the generator of the semigroup $(U_{t,\epsilon}^{(n)})_{t \geq 0}$.  Indeed, for $[A_{i,j}]_{i,j=1, \dots n} \in D(L) \otimes M_n( \C)$,
$$(1- \epsilon L^{(n)})([A_{i,j}]_{i,j=1, \dots n})=[(1-\epsilon L)(A_{i,j})]_{i,j=1, \dots n}$$
hence
$$((1-\epsilon L)^{-1})^{(n)}(1- \epsilon L^{(n)})([A_{i,j}]_{i,j=1, \dots n})=[(1- \epsilon L)^{-1}(1-\epsilon L)(A_{i,j})]_{i,j=1, \dots n}=[A_{i,j}]_{i,j=1, \dots n} ,$$
which proves that  $(1-\epsilon L^{(n)})^{-1}=((1-\epsilon L)^{-1})^{(n)}$.  Hence,
$$L^{(n)}(1-\epsilon L^{(n)})^{-1}=L^{(n)}((1-\epsilon L)^{-1})^{(n)}=(L(1-\epsilon L)^{-1})^{(n)} .$$
Therefore $U_{t, \epsilon}$ is completely positive for all $t \geq 0$ and $\epsilon >0$.  Then, by Theorem \ref{2.2}, there exists a completely positive map $\phi_{\epsilon}$ and $G_{\epsilon} \in \mathcal{B}(\hi)$ such that
$$L_{\epsilon}(A)= \phi_{\epsilon}(A)+G_{\epsilon}A + AG_{\epsilon}^*$$
for all $A \in \mathcal{B}(\hi)$.  Next, we claim that $L_{\epsilon}(A) \underset{\epsilon \to 0}{\longrightarrow} L(A)$ in the $\sigma$-weak topology for all $A \in D(L)$.  Let $A \in \mathcal{B}(\hi)$.  First, we want to show $(1-\epsilon L)^{-1}(A) \underset{\epsilon \to 0}{\longrightarrow} A$ $\sigma$-weakly so let $\eta$ be an element of the predual $L_1( \hi)$ of $\mathcal{B}(\hi)$ and $\gamma >0$.  Since $U_t(A) \underset{t \to 0}{\longrightarrow} A$ $\sigma$-weakly, choose $\delta>0$ so that for any $t < \delta$ we have $| \eta (U_t(A)-A)| < \gamma / 2$.  Hence 
$$\int_{0}^{\delta} \epsilon^{-1}e^{-t/ \epsilon} \left| \eta (U_t(A)-A) \right| dt <\frac{\gamma}{2} .$$
Then,
\begin{align*}
\left| \eta \left( (1- \epsilon L)^{-1}(A) \right) -\eta \left( A \right) \right| & =  \left| \eta \left(\epsilon^{-1} (\epsilon^{-1}- L)^{-1}(A) \right) -\eta \left( A \right) \right|
\\ & = \left| \int_0^{\infty} \epsilon^{-1}e^{-t/ \epsilon} \eta (U_t(A)) dt - \eta (A) \right| \quad \text{ by \cite[Prop.3.1.6]{br}}
\\ & \leq \int_{\delta}^{\infty} \epsilon^{-1}e^{-t/ \epsilon} \left| \eta (U_t(A)-A) \right| dt + \int_{0}^{\delta} \epsilon^{-1}e^{-t/ \epsilon} \left| \eta (U_t(A)-A) \right| dt
\\ & \leq 2 \| \eta \| \| A \| \int_{\delta}^{\infty} \epsilon^{-1}e^{-t/ \epsilon}dt + \frac{\gamma}{2}
\\ & = 2 \| \eta \| \|A \|e^{-\delta / \epsilon}+ \frac{\gamma}{2} .
\end{align*}
So pick $\epsilon_0 > 0 $ so that for all $0 < \epsilon < \epsilon_0$ we have $e^{-\delta / \epsilon} < \gamma (4\| \eta \| \| A \| )^{-1}$.  Then we have that $\left| \eta \left( (1- \epsilon L)^{-1}(A) \right) -\eta \left( A \right) \right|< \gamma$ and therefore $(1-\epsilon L)^{-1}(A) \underset{\epsilon \to 0}{\longrightarrow} A$ $\sigma$-weakly for all $A \in \mathcal{B}(\hi)$.  So, for $A \in D(L)$, replace A with $LA$ and we then have $L(1- \epsilon L)^{-1}A \underset{\epsilon \to 0}{\longrightarrow} LA$ $\sigma$-weakly since $L(1- \epsilon L)^{-1}A=(1- \epsilon L)^{-1}LA$ for any $A \in D(L)$.  Hence $L_{\epsilon}(A) \underset{\epsilon \to 0}{\longrightarrow} L(A)$ $\sigma$-weakly for all $A \in D(L)$.  Thus

$$L(A)=\sigma \text{-weak-}\lim_{\epsilon \rightarrow 0} \left( \phi_{\epsilon}(A)+G_{\epsilon}A+AG_{\epsilon}^* \right)$$
which completes the proof.  
\end{proof}

In Theorem \ref{3.1}, if $A \in D(L^2)$ we actually get that $L_{\epsilon}(A) \underset{\epsilon \to 0^+}{\longrightarrow} L(A)$ in norm.  Indeed, for $A \in D(L)$,
$$ \| (1- \epsilon L)^{-1}A-A \| = \|((1- \epsilon L)^{-1}-(1- \epsilon L)^{-1}(1- \epsilon L))A \|= \epsilon \| (1- \epsilon L)^{-1}LA \| \leq \epsilon \| LA \|$$
since $\|(1- \epsilon L)^{-1} \| \leq 1$ for every $ \epsilon > 0$ (see \cite[Prop. 3.1.10]{br}).  So, for $A \in D(L)$
$$ \| (1- \epsilon L)^{-1}A - A \| \leq \epsilon \| LA \| \rightarrow 0$$
as $\epsilon \rightarrow 0$.  Hence, if $A \in D(L^2)$ then
$$L_{\epsilon}(A)=L(1- \epsilon L)^{-1}A=(1- \epsilon L)^{-1}LA \underset{\epsilon \to 0}{\longrightarrow} LA .$$

For a general QMS on $\mathcal{B}(\hi)$, we would not expect the completely positive part of the representation of the generator to be bounded.  This leads us to the following definition:

\begin{deff}\label{3.3}
Let $U$ be a subspace of a Hilbert space $\hi$.  A linear map $\phi$ from a linear subspace $\mathcal{A}$ of $\mathcal{B} (\hi)$ to the set of sesquilinear forms on $U \times U$ is $\bf{U-completely}$ $\bf{positive}$ if for any $k \in \N$, any positive operator $A=(A_{i,j})_{i,j =1, \dots , k} \in \mathcal{A} \otimes M_k(\C)$ and for all $u_1, \dots , u_k \in U$ we have that
$$\sum_{i,j=1}^k \phi(A_{i,j})(u_i,u_j) \geq 0 .$$
\end{deff}

We now proceed to give analogous forms to Lindblad's for the generator of a QMS.

\begin{thm}\label{3.6}
Let $L$ be the generator of a QMS on the von Neumann algebra $ \mathcal{B}(\hi)$.  Then there exists a linear (not necessarily closed) subspace $W$ of $\hi$, a $W$-completely positive map $\phi$ from $D(L)$ into the set of sesquilinear forms on $W \times W$, and a linear operator K from W to $\hi$ such that
$$\la u, L(A)v \ra= \phi (A)(u,v)+ \la Ku , Av \ra +\la u, AKv \ra$$
for all $A \in D(L)$ and all $u,v \in W$. 
\end{thm}

\begin{proof}
By Proposition \ref{3.1} there exists a family $(\phi_{\epsilon})_{\epsilon >0}$ of normal completely positive operators on $\mathcal{B} ( \hi)$ and there exists a family $(G_{\epsilon})_{\epsilon >0} \subseteq \mathcal{B}(\hi)$ such that
$$L(A)=\lim_{\epsilon \rightarrow 0} \left( \phi_{\epsilon}(A)+G_{\epsilon}A+AG_{\epsilon}^* \right)$$
for all $A \in D(L)$ where the limit is taken in the $\sigma$-weak topology.  Define $W \subseteq \hi$ by
$$W = \{ u \in \hi : \lim_{\epsilon \rightarrow 0} \la h, G_{\epsilon}^*u \ra \text{ exists for all } h \in \hi \} .$$
Then define $K$ on W by $Ku=\text{weak-} \lim_{\epsilon \rightarrow 0 } G_{\epsilon}^*u$.  Then, for $A \in D(L)$, 
\begin{align*}
\la u, L(A)v \ra & = \lim_{\epsilon \rightarrow 0} \la u , \left( \phi_{\epsilon}(A)+G_{\epsilon}A + AG_{\epsilon}^* \right) v \ra
\\ & =\lim_{\epsilon \rightarrow 0} \la u, \phi_{\epsilon}(A) v \ra +\la Ku,Av \ra + \la u, AKv \ra
\end{align*}
for all $u,v \in W$.  Further, since $\lim_{\epsilon \rightarrow 0} \la u, \phi_{\epsilon}(A)v \ra$ exists for all $A \in D(L)$ and for all $u,v \in W$, define a linear map $\phi$ from $D(L)$ to the sesquilinear forms on $W \times W$ by
$$\phi(A)(u,v)= \lim_{\epsilon \rightarrow 0} \la u, \phi_{\epsilon}(A)v \ra .$$
Let $A=(A_{i,j})_{i,j= 1, \dots , k} \in D(L) \otimes M_k(\C)$ be a positive operator and let $u_1, \dots , u_k \in W$.  Since $\phi_{\epsilon}$ is completely positive, we have that
$$ \sum_{i,j=1}^k \la u_i, \phi_{\epsilon}(A_{i,j})u_j \ra \geq 0 .$$
Since $\la u ,\phi_{\epsilon}(A)v \ra \underset{\epsilon \to 0}{\longrightarrow} \phi (A)(u,v)$ for all $A \in D(L)$ and $u,v \in W$ we have that
$$\sum_{i,j=1}^n \phi(A_{i,j})(u_i,u_j) \geq 0$$
which proves that $\phi$ is W-completely positive.  
\end{proof}

\begin{rmk}\label{remark3.6}
Assume L is the generator of a QMS on $\mathcal{B}(\hi)$, $T \in D(L)$ is a positive finite rank operator and h is an associate vector for T.  Assume also that $|x \ra \la Th| \in D(L)$, for all $x \in \hi$.  Let the operators $(G_{\epsilon})_{\epsilon >0}$ be defined as in the proof of Theorem \ref{3.6}.  Then $\text{weak-} \lim_{\epsilon \rightarrow 0}G_{\epsilon}(x)$ and $\text{weak-} \lim_{\epsilon \rightarrow 0}G_{\epsilon}^*(x)$ exist for all $x \in \hi$.  Hence the conclusion of Theorem \ref{3.6} is valid with $W= \hi$.  
\end{rmk}

\begin{proof}
For $x \in \hi$, to see that $\text{weak-} \lim_{\epsilon \rightarrow 0}G_{\epsilon}(x)$ exists notice that for all $y \in \hi$
$$ \la y, G_{\epsilon}(x) \ra = \la y , L_{\epsilon}(|x \ra \la Th|)h \ra - \frac{1}{2} \la h , L_{\epsilon}(T)h \ra \la y,x \ra .$$
Since $T, |x \ra \la Th | \in D(L)$, we have by Theorem \ref{3.1} that $L_{\epsilon}(T) \underset{\epsilon \to 0}{\longrightarrow} L(T)$ and $L_{\epsilon}(|x\ra \la Th|) \underset{\epsilon \to 0}{\longrightarrow} L(|x\ra \la Th|)$ $\sigma$-weakly.  Thus $\la h, L_{\epsilon}(T)h \ra \underset{\epsilon \to 0}{\longrightarrow} \la h, L(T) \ra$ and $\la h, L_{\epsilon}(|x \ra \la Th|)h \ra \underset{\epsilon \to 0}{\longrightarrow} \la h, L(|x \ra \la Th|) \ra$.  Hence $G_{\epsilon}(x) \underset{\epsilon \to 0}{\longrightarrow} L(|x \ra \la Th|)h - \frac{1}{2} \la h , L(T)h \ra x$ weakly.  Next, to see that $\text{weak-} \lim_{\epsilon \rightarrow 0}G_{\epsilon}^*(x)$ exists for all $y \in \hi$, notice that for all $x \in \hi$, $\la G_{\epsilon}^*(y) , x \ra = \la y , G_{\epsilon}(x) \ra$.
\end{proof}
Note that Theorem \ref{3.6} does not specify the size of the subspace W, while Remark \ref{remark3.6} guarantees that $W= \hi$ under some rather strong assumptions.  Theorem \ref{3.9} gives a form of the generator similar to that of Theorem \ref{3.6} with the added advantage that the subspace W is replaced by a subspace U which is easy to describe.  The easy form of U enables us to verify that it is dense in $\hi$ in Examples \ref{4.1} and \ref{4.2}.

\begin{deff}\label{domainalgebra}
If L is the generator of a QMS then the $\bf{domain}$ $\bf{algebra}$ of L is the largest $^*$-subalgebra of the domain of L, $D(L)$, and is shown in \cite{arveson} to be given by
$$\mathcal{A}= \{A \in D(L):A^*A, AA^* \in D(L) \} .$$
\end{deff}

\begin{thm}\label{3.9}
Let L be the generator of a QMS on $\mathcal{B}(\hi)$.  Let $D(L)$ denote its domain and $\mathcal{A}$ denote its domain algebra.  Assume there exists a positive finite rank operator T in $D(L)$ and an associate vector h for T such that $|Th\ra \la Th| \in D(L)$.  Let U be the linear subspace of $\hi$ defined by $U= \{ x \in \hi : |x \ra \la Th| \in \mathcal{A} \}$ and let $G:U \rightarrow \hi$ be the linear operator defined by
$$G(u)=L(|u \ra \la Th| )h- \frac{1}{2} \la h , L(T)h \ra u .$$
Then there exists a U-completely positive map $\phi$ from $\mathcal{A}$ to the set of sesquilinear forms on $U \times U$ such that
$$ \la u, L(A)v \ra = \phi (A)(u,v) + \la u, GA v \ra + \la GA^*u,v \ra .$$
for all $A \in \mathcal{A}$ and $u,v \in U$.
\end{thm}

\begin{rmk}\label{remark2}
First, for the sake of clarity we explain the definition of U.  Note that by Definition \ref{domainalgebra}, for $x \in \hi$, $|x \ra \la Th| \in \mathcal{A}$ is equivalent to having the following three conditions hold:  $|x \ra \la Th| \in D(L)$, $(|x \ra \la Th|)^* \circ |x \ra \la Th|= \| x \|^2 | Th \ra \la Th| \in D(L)$, and $|x \ra \la Th| \circ (|x \ra \la Th|)^*= \| Th \|^2 |x \ra \la x| \in D(L)$.  Thus if U contains non-zero vectors then $|Th \ra \la Th| \in D(L)$ and that is why this condition appears explicitly in the statement of Theorem \ref{3.9}.
\end{rmk}

\begin{proof}[Proof of Theorem \ref{3.9}]
By Theorem \ref{3.1} there exists a family $(L_{\epsilon})_{\epsilon >0}$ of generators of uniformly continuous QMSs on $\mathcal{B}(\hi)$ such that $L(A)= \sigma \text{-weak-}\lim_{\epsilon \rightarrow 0^+}L_{\epsilon}(A)$ for every $A \in D(L)$.  Also there exist families of completely positive operators $(\phi_{\epsilon})_{\epsilon >0}$ on $\mathcal{B}(\hi)$ and bounded operators $(G_{\epsilon})_{\epsilon >0}$ on $\hi$ such that
$$L_{\epsilon}(A)= \phi_{\epsilon}(A) +G_{\epsilon}A+AG_{\epsilon}^*$$
for all $A \in D(L)$.  Let $v \in U$ and let $A \in \mathcal{A}$.  Since $\mathcal{A}$ is an algebra, we obtain $|Av \ra \la Th|=A \circ |v \ra \la Th| \in \mathcal{A}$. Then, using the explicit form for $G_{\epsilon}$ from Theorem \ref{2.2}, we have
\begin{equation}\label{3.9.1}
G_{\epsilon}Av=L_{\epsilon}(|Av \ra \la Th |)h - \frac{1}{2} \la h, L_{\epsilon}(T)h \ra Av .
\end{equation}
Since $ |Av \ra \la Th| \in \mathcal{A} \subseteq D(L)$ we obtain by Theorem \ref{3.1} that
$L_{\epsilon}(|Av \ra \la Th|) \underset{\epsilon \to 0}{\longrightarrow} L(|Av \ra \la Th|)$ in the $\sigma$-weak topology.  Thus for any $u \in \hi$ we obtain
\begin{equation}\label{3.9.2}
\la u, L_{\epsilon}( |Av \ra \la Th|)h \ra \underset{\epsilon \to 0}{\longrightarrow}  \la u, L( |Av \ra \la Th|)h \ra .
\end{equation}
Also, by Theorem \ref{3.1}, since $T \in D(L)$ we have that $L_{\epsilon}(T) \underset{\epsilon \to 0}{\longrightarrow} L(T)$ in the $\sigma$-weak topology and hence
\begin{equation}\label{3.9.3}
\la h, L_{\epsilon} (T)h \ra \underset{\epsilon \to 0}{\longrightarrow} \la h , L(T)h \ra .
\end{equation}
Thus, by \eqref{3.9.1}, \eqref{3.9.2}, and \eqref{3.9.3}, for any $u \in \hi$, $v \in U$ and $A \in \mathcal{A}$ we have
\begin{align*}
\la u , G_{\epsilon}Av \ra & = \la u , L_{\epsilon}(|Av \ra \la Th|)h \ra - \frac{1}{2} \la h, L_{\epsilon}(T)h \ra \la u, Av \ra
\\ & \underset{\epsilon \to 0}{\longrightarrow}  \la u , L(|Av \ra \la Th|)h \ra - \frac{1}{2} \la h, L(T)h \ra \la u, Av \ra = \la u , GAv \ra .
\end{align*}
Similarly, for $u \in U$, $v \in \hi$, and $A \in \mathcal{A}$, we have 
$$ \la u, AG_{\epsilon}^*v \ra \underset{\epsilon \to 0}{\longrightarrow} \la GA^*u,v \ra .$$
Thus for $u,v \in U$ and $A \in \mathcal{A}$,
\begin{align*}
\la u, L(A)v \ra & = \lim_{\epsilon \rightarrow 0} \la u , L_{\epsilon}(A)v \ra = \lim_{\epsilon \rightarrow 0} \la u , \left( \phi_{\epsilon}(A)+G_{\epsilon}A + AG_{\epsilon}^* \right) v \ra
\\ & =\lim_{\epsilon \rightarrow 0} \la u, \phi_{\epsilon}(A) v \ra +\la u,GAv \ra + \la GA^*u, v \ra .
\end{align*}
Thus $\lim_{\epsilon \rightarrow 0} \la u, \phi_{\epsilon}(A)v \ra$ exists for all $A \in \mathcal{A}$ and for all $u,v \in U$,  and therefore define
$$\phi(A)(u,v)= \lim_{\epsilon \rightarrow 0} \la u, \phi_{\epsilon}(A)v \ra .$$
Let $A=(A_{i,j})_{i,j= 1, \dots , k} \in \mathcal{A} \otimes M_k(\C)$ be a positive operator and let $u_1, \dots , u_k \in U$.  Since $\phi_{\epsilon}$ is completely positive we have that
$$ \sum_{i,j=1}^k \la u_i, \phi_{\epsilon}(A_{i,j})u_j \ra \geq 0 .$$
Since $\la u ,\phi_{\epsilon}(A)v \ra \rightarrow \phi (A)(u,v)$ for all $A \in \mathcal{A}$ and $u,v \in U$ we have that
$$\sum_{i,j=1}^n \phi(A_{i,j})(u_i,u_j) \geq 0 .$$
Therefore $\phi$ is U-completely positive. 
\end{proof}

While restricting to the domain algebra helps us to understand the subspace U and the operator G, it does come at a cost since the domain of the generator is $\sigma$-weakly dense while there are examples of QMSs whose domain algebras are not very large.  Indeed, in \cite{fagnola2}, F. Fagnola gives an example of a QMS on $\mathcal{B}(L_2(0, \infty), \C )$ where $\mathcal{A}$ is not $\sigma$-weakly dense in $\mathcal{B}(L_2(0, \infty), \C )$.  In Section 6 we will look at several examples where U is dense in $\hi$ and also verify the above form for the generator L.

We will proceed by showing that we have analogous results to that of Stinespring's.  In the next proposition when we say a map $\pi: \mathcal{A} \rightarrow \mathcal{B}( \hi)$, where $\mathcal{A}$ is a (not necessarily closed) unital $^*$-subalgebra of $\mathcal{B}( \hi)$, is a unital $^*$-representation we mean that it is a unital norm-continuous $^*$-homomorphism.

\begin{thm}\label{3.4}
Suppose $\mathcal{A}$ is a unital (not necessarily closed) $^*$-subalgebra of $\mathcal{B} (\hi)$, $U$ is a (not necessarily closed) linear subspace of $\hi$, and $\phi$ is a U-completely positive map from $\mathcal{A}$ to the set of sesquilinear forms on $U \times U$.  Then there exists a Hilbert space $\mathcal{K}$, a unital $^*$-representation $\pi: \mathcal{A} \rightarrow \mathcal{B} (\mathcal{K})$ of norm equal to one, and a linear map $V: U\rightarrow \mathcal{K}$ such that
$$  \phi(A)(u,w)= \la Vu, \pi (A)Vw \ra_{_{\mathcal{K}}}$$
for all $u,w \in U$.
\end{thm}
\begin{proof}
Define a sesquilinear form $( \cdot , \cdot ): (\mathcal{A} \otimes U ) \times (\mathcal{A} \otimes U) \rightarrow \C$ by
$$(x,y)= \sum_{i,j=1}^n \phi(A_i^*B_j)(u_i,v_j)$$
where $x=\sum_{i=1}^n A_i \otimes u_i$ and $y=\sum_{j=1}^n B_j \otimes v_j$ (since we allow zero entries, we can have the same upper limit n in both sums).  Since $\phi$ is U-completely positive, $(x,x)\geq 0$ for all $x \in \mathcal{A} \otimes U$ so $( \cdot , \cdot )$ is a positive definite sesquilinear form.  For $x \in \mathcal{A} \otimes U$ let $\|x \|_{(\cdot , \cdot)}=\sqrt{(x,x)}$.  Let $N= \{ x \in \mathcal{A} \otimes U: (x,x)=0 \}$.  Since $( \cdot , \cdot )$ is a positive definite sesquilinear form, by the Cauchy-Schwartz inequality, N is a linear subspace of $\mathcal{A} \otimes U$ and we have that the completion of $( \mathcal{A} \otimes U) / N$, which we'll denote by $\mathcal{K}$, is a Hilbert space where the inner product is given by $\la x+N,y+N \ra_{_{\mathcal{K}}}=(x,y)$.  Let $\pi_0:\mathcal{A} \rightarrow L(\mathcal{A} \otimes U)$ (where $L(X)$ denotes the linear (not necessarily bounded) operators from X to X) defined by
$$\pi_0(A) \left( \sum_{i=1}^n A_i \otimes u_i \right) = \sum_{i=1}^n AA_i \otimes u_i .$$
Then, for $A \in \mathcal{A}$, $x=\sum_{i=1}^n A_i \otimes u_i \in \mathcal{A} \otimes U$ and $y=\sum_{j=1}^n B_j \otimes v_j \in \mathcal{A} \otimes U$ we have
$$(x, \pi_0(A)y)=\left( \sum_{i=1}^n A_j \otimes u_j , \sum_{j=1}^n AB_j \otimes v_j \right)=\sum_{i,j=1}^n \phi((A^*A_i)^*B_j)(u_i,v_i) = (\pi_0(A^*)x,y) .$$
Fix $x=\sum_{i=1}^n A_i \otimes u_i \in \mathcal{A} \otimes U$ and define $\omega : \mathcal{A} \rightarrow \C$ by $\omega(A)=(x, \pi_0(A)x)$ for $A \in \mathcal{A}$.  Clearly $\omega$ is linear.  Then for $A \in \mathcal{A}$,
$$\omega(A^*A) = (x, \pi_0(A^*A)x)= \sum_{i,j=1}^n \phi(A_i^*A^*AA_j)(u_i,u_j) = \sum_{i,j=1}^n \phi((AA_i)^*(AA_j))(u_i,u_j) \geq 0$$
since $\phi$ is U-completely positive.  Now, for $A \in \mathcal{A}$, $A^*A \leq \|A^*A \|1$ since $1 \in \mathcal{A}$.  Then, since $\omega$ is positive,
$$\omega(A^*A) \leq \|A^*A \| \omega(1)= \|A^*A \| \|x \|_{_{( \cdot, \cdot)}}^2 .$$
So,
\begin{equation}\label{s1}
  \| \pi_0(A)x \|_{_{( \cdot, \cdot)}}^2 = ( \pi_0(A)x, \pi_0(A)x )=(x, \pi_0(A^*A)x)= \omega (A^*A) \leq \|A^*A \| \|x \|_{_{( \cdot, \cdot)}}^2 .
\end{equation}
Thus in fact $\pi_0(A) \in \mathcal{B}(\mathcal{A} \otimes U )$ (bounded operators from $\mathcal{A} \otimes U$ to $\mathcal{A} \otimes U$) and $\| \pi_0(A) \| \leq \sqrt{\|A^*A \|} = \|A \|$.
Hence, if $(x,x)=0$ then $(\pi_0(A)x , \pi_0(A)x) = 0$ for all $A \in \mathcal{A}$.
Now, define $\pi: \mathcal{A} \rightarrow \mathcal{B} (\mathcal{K})$ by $\pi(A)(x+N)= \pi_0(A)x+N$ which is well-defined since we saw above that $(x,x)=0 \Rightarrow (\pi_0(A)x, \pi_0(A)x)=0$.  It is obvious that $\pi$ is linear, $\pi (1_{\hi})=1_{\ki}$, and for $A,B \in \mathcal{A}$ we have $\pi(A^*)= \pi (A)^*$ and $\pi (AB)= \pi (A) \pi (B)$ as in Stinesprings's proof \cite[Theorem 1]{stinespring}.  Further, let $V:U \rightarrow \mathcal{K}$ where $Vu= 1 \otimes u +N$ for all $u \in U$.  Then, for $u,w \in U$ and $A \in \mathcal{A}$ we have that
$$\la Vu, \pi(A)Vw\ra_{_{\mathcal{K}}}= \la 1 \otimes u +N, A \otimes w +N\ra_{_{\mathcal{K}}}=(1 \otimes u, A \otimes w)=\phi(A)(u,w) .$$
Any representation of a unital $C^*$-algebra into another is known to be bounded and in fact have norm equal to one (obtained at the identity) \cite[Lemma 3.4.2(b)]{sunder}.  The domain algebra $\mathcal{A}$ is not closed so it is not a $C^*$-algebra but we verify here that the representation $\pi$ has norm equal to one.
Indeed, let $A \in \mathcal{A}$ and $x+N \in \ki$.  Then
$$\| \pi (A)(x+N) \|_{\ki} = \| \pi_0 (A)x+N \|_{\ki} = \la \pi_0 (A)x+N, \pi_0 (A)x+N \ra^{1/2} =\| \pi_0(A)x \|_{( \cdot, \cdot)}$$
Further,  by \eqref{s1}, 
$$\| \pi_0(A)x \|_{( \cdot, \cdot)} \leq \|A^*A\|^{1/2} \| x \|_{( \cdot, \cdot)} = \| A \| \|x+N \|_{\ki} .$$
and therefore $\| \pi (A) \| \leq \| A \|$ for all $A \in \mathcal{A}$ and the proof is complete.  
\end{proof}

Theorems \ref{3.9} and \ref{3.4} bring us one step closer to the explicit form of the generator of a QMS.  Our progress is summed up in the following which is the main result of our paper.

\begin{cor}\label{3.7}
Let $L$ be the generator of a QMS on the von Neumann algebra $\mathcal{B} (\hi)$ and let $\mathcal{A}$ be its domain algebra.  Suppose there exists a positive finite rank operator $T \in D(L)$ and an associate vector h for T such that $|Th\ra \la Th| \in D(L)$.  Let U be the linear subspace of $\hi$ defined by $U= \{ x \in \hi: |x \ra \la Th| \in \mathcal{A} \}$.  Then there exists a Hilbert space $\ki$, a unital $^*$-representation $\pi : \mathcal{A} \rightarrow \mathcal{B}(\ki)$, and linear maps $G: U \rightarrow \hi$ and $V:U \rightarrow \ki$ such that
$$\la u, L(A)w \ra = \la Vu, \pi (A)Vw \ra_{_{\mathcal{K}}}+ \la u,GAw \ra + \la GA^*u,w \ra$$
for all $u,w \in U$ and $A \in \mathcal{A}$.
\end{cor}

\begin{proof}
Follows immediately from \ref{3.9} and \ref{3.4}.  
\end{proof}

We do not know if the map G that appears in Theorem \ref{3.9} and Corollary \ref{3.7} is closed.  In Proposition \ref{3.10} we define a linear operator $\widehat{G}:U \rightarrow \mathcal{B}(\hi)$ such that $\widehat{G}(x)(h)=G(x)$, for  a positive finite rank operator $T \in D(L)$ and h an associate vector of T, and we study its closability.

\begin{prop}\label{3.10}
Let L be the generator of a QMS on the von Neumann algebra $\mathcal{B} (\hi)$.  Suppose there exists a positive, finite rank operator $T \in D(L)$ and an associate vector h for T such that $|Th \ra \la Th| \in D(L)$.  Let $U= \{ x \in \hi : |x \ra \la Th | \in \mathcal{A} \}$ and define $\widehat{G}: U \rightarrow \mathcal{B}(\hi)$ by
$$\widehat{G}(x)(v)=L( |x \ra \la Th | )v -\frac{1}{2} \la h, L(T)v \ra x .$$
Then $\widehat{G}$ is ($\| \cdot \|$, $\sigma$-weakly)-closable.  Further, if we define $U_0= \{ x \in \hi : |x \ra \la Th| \in D(L) \}$ then $U \subseteq U_0$ and $\widehat{G}$ defined on $U_0$ is  ($\| \cdot \|$, $\sigma$-weakly)-closed.
\end{prop}

\begin{proof}
Let $(x_n)_{n \geq 1} \subseteq U$ such that $x_n \rightarrow 0$ in norm and $\widehat{G} (x_n) \rightarrow A \in \mathcal{B}( \hi)$ $\sigma$-weakly. Then, $|x_n\ra \la Th| \in \mathcal{A} \subseteq D(L)$.  We claim that $|x_n \ra \la Th| \underset{\epsilon \to 0}{\longrightarrow} 0$ $\sigma$-weakly.  Indeed, let $(u_k)_{k \geq 1}, (v_k)_{k \geq 1} \subseteq \hi$ such that $\sum_k \|u_k \|^2< \infty$ and $\sum_k \|v_k \|^2< \infty$.  Then
\begin{align*}
\left| \sum_{k=1}^{\infty} \la u_k, |x_n \ra \la Th|v_k \ra \right| & \leq \sum_{k=1}^{\infty} \left| \la Th, v_k \ra \la u_k , x_n \ra \right| 
\\ & \leq \|x_n \| \left(\sum_{k=1}^{\infty} \|Th \|^2 \|v_k \|^2 \right)^{1/2} \left( \sum_{k=1}^{\infty} \|u_k \|^2 \right)^{1/2}
\\ & =c_1 \|x_n \|
\end{align*}
and since $\| x_n \| \rightarrow 0$ we have that $|x_n \ra \la Th | \rightarrow 0$ $\sigma$-weakly.  Similarly, we claim that the sequence of bounded linear operators $v \mapsto \la h , L(T)v \ra x_n$ (simply denoted as $\la h, L(T) \cdot \ra x_n$) converges to 0 $\sigma$-weakly as $n \to \infty$.  Indeed,
$$\left| \sum_{k=1}^{\infty} \la u_k, \la h, L(T)v_k \ra x_n \ra \right| \leq c_2 \| x_n \| \left( \sum_{k=1}^{\infty} \|u_k \|^2 \right)^{1/2}\left( \sum_{k=1}^{\infty} \|v_k \|^2 \right)^{1/2} .$$
Since $\la h , L(T) \cdot \ra x_n \rightarrow 0$ $\sigma$-weakly as $n \to \infty$ and $\widehat{G}(x_n) \to A$ $\sigma$-weakly we have that $L(|x_n \ra \la Th|) \rightarrow A$ $\sigma$-weakly.  Thus, since $L$ is $\sigma$-weakly closed on its domain D(L) \cite[Theorem 3.1.10]{br}, and $|x_n \ra \la Th| \rightarrow 0$ $\sigma$-weakly we have that $A=L(0)=0$ and therefore $\widehat{G}$ is closable.  For the last statement of Proposition \ref{3.10} suppose that $(x_n)_{n \geq 1} \subseteq U_0$ with $x_n \rightarrow x$ in norm and $\widehat{G}(x_n) \rightarrow A \in \mathcal{B}(\hi)$ $\sigma$-weakly.  Repeating the above argument with $x_n-x$ in place of $x_n$ we obtain that $|x_n-x \ra \la Th| \rightarrow 0$ $\sigma$-weakly (hence $|x_n \ra \la Th| \rightarrow |x \ra \la Th|$ $\sigma$-weakly), and that $\la h , L(T) \cdot \ra (x_n-x) \rightarrow 0$ $\sigma$-weakly as $n \to \infty$, hence $ \la h, L(T) \cdot \ra x_n \rightarrow \la h , L(T) \cdot \ra x$ $\sigma$-weakly as $n \rightarrow \infty$.  Since $\widehat{G}(x_n) \rightarrow A$ $\sigma$-weakly, we obtain that
$$L(|x_n \ra \la Th|) \rightarrow A+ \frac{1}{2} \la h , L(T) \cdot \ra x$$
$\sigma$-weakly.  Thus, since L is $\sigma$-weakly closed on its domain $D(L)$, we obtain that $|x \ra \la Th|)=A+ \frac{1}{2} \la h , L(T) \cdot \ra x$, i.e., $\widehat{G}(x)=A$, which proves that $\widehat{G}$ defined on $U_0$ is ($ \| \cdot \|$, $\sigma$-weakly)-closed.
\end{proof}

As mentioned earlier, we will illustrate the form of the generator L and discuss the subspace U in several examples in Section \ref{sn6} but first we would like to attempt obtaining an analogous result to that of Kraus' (Theorem \ref{2.4}).

\section{An Attempt to Extend Kraus' Result}\label{attemptkraus}
Theorems \ref{3.9} and \ref{3.4} describe the form of the generator of a QMS on $\mathcal{B}(\hi)$.  Unfortunately we do not have a result similar to Theorem \ref{2.4} for the form of the representation $\pi: \mathcal{A} \rightarrow \mathcal{B}(\hi)$ which appears in the conclusion of Theorem \ref{3.4}.  For the uniformly continuous QMSs on $\mathcal{B}(\hi)$, $\pi$ turns out to be a normal representation on $\mathcal{B}(\hi)$ and the map $V$ which appears in Theorem~\ref{2.3} turns out to be bounded.

This section is dedicated to proving, under suitable assumptions, continuity properties of the operators V and $\phi$  which appear in Theorem \ref{3.4} in the hopes of obtaining a dilation for $\phi$, similar to Theorem \ref{2.4}.  While we do not achieve this, we get rather close and identify what we see is ultimately needed to finish.  We also have some continuity results which are of interest in their own right.  

\begin{prop}\label{5.1}
Let L be the generator of a QMS on the von Neumann algebra $\mathcal{B} (\hi)$.  Further, suppose there exists a positive, finite rank operator $T \in D(L)$ and an associate vector h for T such that $|Th\ra \la Th| \in D(L)$.  Let $U= \{ x \in \hi : |x \ra \la Th | \in \mathcal{A} \}$ and define $G: U \rightarrow \hi$ by
$$Gx= L( |x \ra \la Th|)h - \frac{1}{2} \la h, L(T)h \ra x$$
and $V:U \rightarrow \ki$ by 
$$Vx=1 \otimes x +N$$
where $\ki$ is the Hilbert space given in Theorem \ref{3.4}.  Also, suppose that 
\begin{equation}\label{*}
\text{there exists } C >0 \text{ such that } \| L(|x \ra \la Th|)h \| \leq C \|x \| \text{ for all } x \in U.
\end{equation}
Then G is bounded on U.  If the map $\phi$ of Theorem \ref{3.4} satisfies the conclusion of Theorem \ref{3.9} then the map V is bounded on U as well.
\end{prop}

\begin{proof}
For $x \in U$,
$$
\|Gx \| = \| L(|x \ra \la Th|)h - \frac{1}{2} \la h , L(T)h \ra x \|
 \leq C \| x \| + \frac{1}{2} \left| \la h , L(T)h \ra \right| \|x \|
 \leq C' \|x \|
$$
and so $G$ is bounded on U.  Further, let $x \in U$.  Then
$$ \|Vx \|_{\ki}^2= \| 1 \otimes x +N \|_{\ki}^2=(1 \otimes x, 1 \otimes x)= \phi(1) (x,x)=|\phi(1)(x,x)| .$$
Hence, by the conclusion of Theorem \ref{3.9}, since $L(1)=0$ we get that
$$|\phi(1)(x,x)|= | - \la x, Gx \ra - \la Gx, x \ra | \leq C \|x \|^2+ C \|x \|^2 .$$
Therefore V is also bounded on U.  
\end{proof}

The operator G of Proposition \ref{5.1} is the same as in Theorem \ref{3.9} and Corollary \ref{3.7}.  The operator V of Proposition \ref{5.1} is the same as in Theorem \ref{3.4} and Corollary \ref{3.7}.  Corollary \ref{3.7} and Proposition \ref{5.1} are used in the proof of the next result.
 
 \begin{prop}\label{5.2}
 Let L be the generator of a QMS on the von Neumann algebra $\mathcal{B} (\hi)$ and let $\mathcal{A}$ denote its domain algebra.  Further, suppose there exists a positive, finite rank operator $T \in D(L)$ and an associate vector h for T such that $|Th\ra \la Th| \in D(L)$.  Let $U= \{ x \in \hi : |x \ra \la Th | \in \mathcal{A} \}$.  Assume that \eqref{*} is valid and that $\overline{U}^{\| \cdot \|} = \hi$.  
Then, there exist a linear map $G:U \rightarrow \hi$, a Hilbert space $\ki$, a linear map $V:U \rightarrow \ki$ and a unital $^*$-representation $\pi : \mathcal{A} \rightarrow \mathcal{B}(\hi)$ such that
\begin{equation}\label{5.2.1}
L(A)=V^* \pi (A)V+GA+AG^*
\end{equation}
for all $A \in \mathcal{A}$.  Further, define $\psi: \mathcal{A} \rightarrow \mathcal{B}(\hi)$ by $\psi (A)=GA+AG^*$.  Then $\psi$ is $\sigma$-weakly - $\sigma$-weakly continuous.  Lastly, the map $\varphi : \mathcal{A} \rightarrow \mathcal{B}(\hi)$ defined by $\varphi(A)=V^* \pi (A)V$ is $\sigma$-weakly - $\sigma$-weakly closable.
\end{prop}

\begin{rmk}
Note that the assumptions of Proposition \ref{5.2} are rather strong since \eqref{5.2.1} implies that L is bounded on $\mathcal{A}$ (but not necessarily on $\mathcal{B}(\hi )$).
\end{rmk}

\begin{proof}[Proof of Prop. \ref{5.2}]
By Corollary \ref{3.7} there exist a linear map $G:U \rightarrow \hi$, a Hilbert space $\ki$, a linear map $V:U \rightarrow \ki$ and a unital $^*$-representation $\pi : \mathcal{A} \rightarrow \mathcal{B}(\hi)$ such that
$$\la x, L(A)y \ra= \la Vx, \pi (A) V y \ra + \la x, GAy \ra + \la GA^*x, y \ra $$
for all $A \in \mathcal{A}$ and $x,y \in U$.  By Proposition \ref{5.1} and the assumption that $\overline{U}^{\| \cdot \|} = \hi$ we see that
$$\la x, L(A)y \ra= \la x, V^* \pi (A) V y \ra + \la x, GAy \ra + \la x, AG^*y \ra$$
for all $A \in \mathcal{A}$ and $x,y \in \hi$.  Thus $L(A)=V^* \pi (A)V+GA+AG^*$ for all $A \in \mathcal{A}$.  Let $(A_{\lambda})_{\lambda} \subseteq \mathcal{B}(\hi)$ be a net such that $A_{\lambda} \underset{\lambda}{\rightarrow} A$ $\sigma$-weakly for some $A \in \mathcal{B}(\hi)$.  Let  $(x_n)_{n \geq 1}, (y_n)_{n \geq 1} \subseteq \hi$ such that  $\sum_{n=1}^{\infty} \|x_n \|^2 < \infty$ and $\sum_{n=1}^{\infty} \|y_n \|^2 < \infty$.  Then
\begin{align*}
\sum_{n=1}^{\infty} \la x_n, \psi(A_{\lambda})y_n \ra & =  \sum_{n=1}^{\infty} \left( \la x_n, GA_{\lambda}y_n \ra + \la x_n, A_{\lambda}G^*y_n \ra \right)
\\ & = \sum_{n=1}^{\infty} \left( \la G^*x_n, A_{\lambda}y_n \ra + \la x_n, A_{\lambda}G^*y_n \ra \right)
\\ & \underset{\lambda}{\rightarrow} \sum_{n=1}^{\infty} \left( \la G^*x_n, Ay_n \ra + \la x_n, AG^*y_n \ra \right)
\end{align*}
since $\sum_{n=1}^{\infty} \|G^*x_n \|^2 < \infty$ and $\sum_{n=1}^{\infty} \|G^*y_n \|^2 < \infty$.  So we have that $\psi$ is $\sigma$-weakly - $\sigma$-weakly continuous.  
\\Next, let $(A_{\lambda})_{\lambda} \subseteq \mathcal{A}$ be a net such that $A_{\lambda} \underset{\lambda}{\rightarrow} 0$ $\sigma$-weakly and $\varphi (A_{\lambda}) \underset{\lambda}{\rightarrow} B$ $\sigma$-weakly, for some $B \in \mathcal{B}(\hi)$, where $\varphi (A)=V^* \pi (A) V$.   Let  $(x_n)_{n \geq 1}, (y_n)_{n \geq 1} \subseteq \hi$ such that  $\sum_{n=1}^{\infty} \|x_n \|^2 < \infty$ and $\sum_{n=1}^{\infty} \|y_n \|^2 < \infty$.  Then,
$$ \sum_{n=1}^{\infty}  \la x_n , L(A_{\lambda})y_n \ra =  \sum_{n=1}^{\infty}\left( \la x_n, V^* \pi (A_{\lambda})Vy_n \ra + \la x_n , \psi (A_{\lambda})y_n \ra \right) \rightarrow  \sum_{n=1}^{\infty}  \la x_n , By_n \ra$$
since $\psi$ is $\sigma$-weakly - $\sigma$-weakly continuous and $\varphi (A_{\lambda}) \rightarrow B$ $\sigma$-weakly.  Then, since L is $\sigma$-weakly-$\sigma$-weakly closed on its domain $D(L)$ and therefore  $\sigma$-weakly-$\sigma$-weakly closable on $\mathcal{A}$ we have that $B=L(0)=0$.  So we have that $\varphi$ is $\sigma$-weakly - $\sigma$-weakly closable.  
\end{proof}

If one assumes \eqref{*} but does not assume that  $\overline{U}^{\| \cdot \|} = \hi$ then the proof of Proposition \ref{5.2} gives the following.
\begin{rmk}
Consider the situation described in Proposition \ref{5.2} without assuming that  $\overline{U}^{\| \cdot \|} = \hi$.  Let
$$F= \left \{ \sum_{n=1}^{\infty} |x_n \ra \la y_n| : (x_n)_{n \geq 1}, (y_n)_{n \geq 1} \subseteq \overline{U}^{\| \cdot \|} \text{ such that } \sum_{n=1}^{\infty} \|x_n \|^2 < \infty \quad , \quad \sum_{n=1}^{\infty} \|y_n \|^2 < \infty \right \}$$
 then $F \subseteq L_1(\hi)$ (the space of trace class operators on $\hi$).  Further, if we define $\psi: \mathcal{A} \rightarrow \mathcal{B}(\hi)$ by $\psi (A)=GA+AG^*$ then under assumption \eqref{*}, $\psi$ is $\sigma (\mathcal{A},F) -\sigma (\mathcal{B}(\hi),F)$ continuous.  
 \end{rmk}
 We did not find an application provided by the above remark.

\begin{deff}\label{5.3}
A pair $(\pi,V)$ satisfying $T(A)=V^* \pi (A)V$ where $\pi$ is a representation on $\mathcal{B}( \ki)$ and $V: \hi \rightarrow \ki$, is called a {\bf Minimal Stinespring Representation} if the set
$$\{ \pi(A)Vu: A \in \mathcal{A}, \quad u \in U \}$$
is total in $\ki$.
\end{deff}

If we look back at Theorem \ref{3.4} to the definitions of $\phi$, $V$, and $\ki$, it is easy to see that our $(\pi,V)$ is a minimal Stinespring representation (as in the proof of the original result of Stinespring \cite[Theorem 1]{stinespring}).  This will be used in the following result.

\begin{prop}\label{5.4}
Let L be the generator of a QMS on the von Neumann algebra $\mathcal{B} (\hi)$ and let $\mathcal{A}$ denote its domain algebra.  Suppose there exists a positive, finite rank operator $T \in D(L)$ and an associate vector h for T such that $|Th\ra \la Th| \in D(L)$.  Let $U= \{ x \in \hi : |x \ra \la Th | \in \mathcal{A} \}$.  Also, suppose that \eqref{*} is valid and that $\overline{U}^{\| \cdot \|} = \hi$.  Then the unital $^*$-representation $\pi : \mathcal{A} \rightarrow \ki$ which appears in the statement of Proposition \ref{5.2}, is $\sigma$-weakly - $\sigma$-weakly closable.
\end{prop}

\begin{proof}
Let $(A_{\lambda})_{\lambda} \subseteq \mathcal{A}$ be a net such that $A_{\lambda} \underset{\lambda}{\longrightarrow} 0$ $\sigma$-weakly and $\pi(A_{\lambda}) \underset{\lambda}{\longrightarrow} B$ $\sigma$-weakly for some $B \in \mathcal{B}(\hi)$.  Let $C,D \in \mathcal{A}$.  Then it is trivial to see that $C^*A_{\lambda}D \underset{\lambda}{\longrightarrow} 0$ $\sigma$-weakly and $\pi(C^*)\pi(A_{\lambda})\pi (D) \underset{\lambda}{\longrightarrow} \pi(C^*)B \pi (D)$ $\sigma$-weakly.  Since, by Proposition \ref{5.1}, V is bounded on $\hi$ we have
$$\varphi(C^*A_{\lambda}D)=V^* \pi ( C^*A_{\lambda}D)V=V^*\pi(C^*)\pi(A_{\lambda})\pi(D)V \underset{\lambda}{\longrightarrow} V^* \pi (C^*)B\pi(D)V$$
$\sigma$-weakly.  Well, $\varphi$ is  $\sigma$-weakly - $\sigma$-weakly closable by Proposition \ref{5.2} and so $V^*\pi(C^*)B\pi(D)V=0$.  Then, for any $x,y \in \hi$
$$ \la \pi(C)Vx,B\pi(D)Vy \ra = \la x, V^* \pi (C^*)B \pi (D)Vy \ra =0 ,$$
and, since $( \pi, V)$ is a minimal representation, $B=0$.  Therefore $\pi$ is closable.  
\end{proof}

In the application of the theorem of Kraus to the generators of uniformly continuous QMSs on $\mathcal{B}(\hi)$, $\pi$ is a  $\sigma$-weakly continuous unital $^*$-representation so, for a cyclic vector $\omega \in \ki$, the map $\mathcal{B}(\hi) \ni A \mapsto \la \omega , \pi (A) \omega \ra$ is positive and $\sigma$-weakly continuous.  Since we have a characterization of such maps, namely positive trace-class operators acting on $\mathcal{B}(\hi)$ via the trace duality, we can conclude this map has the form
$$ \la \omega , \pi (A) \omega \ra = \sum_{n=1}^{\infty} \la x_n , A x_n \ra$$
where $\sum_{n=1}^{\infty} \| x_n \|^2 < \infty$.  Unfortunately, if we replace $\mathcal{B}(\hi)$ with a (not necessarily closed) $^*$-subalgebra $\mathcal{A}$ and we only assume that the unital $^*$-representation $\pi : \mathcal{A} \rightarrow \ki$ is ($\sigma$-weakly, $\sigma$-weakly)-closable (which is guaranteed by Proposition \ref{5.4}) we do not know the form of the map $\mathcal{A} \ni A \mapsto \la \omega, \pi (A) \omega \ra$.  This seems to be the missing ingredient in order to obtain an analogue result of Kraus for general QMS on $\mathcal{B}(\hi)$.

\section{Examples}\label{sn6}

We will now proceed to look at three examples of QMSs where we verify that their generators satisfy the form given by Corollary \ref{3.7}.  We identify the linear maps G, V, the representation $\pi$, the Hilbert space $\ki$ and the linear subspace U of $\hi$ as in Corollary~\ref{3.7}.  Moreover we prove that the subspace U is dense in $\hi$ in the first two examples.

\begin{exm}\label{4.1}
{\bf (Heat Flow \cite{arveson})} Define $P=\frac{1}{i} \frac{d}{dx}$ and $Q$ to be multiplication by $x$ where $P$ and $Q$ act on $L_2(\R)$.  Further, for $A \in \mathcal{B}(L_2(\R))$ define
$$D_P(A)=i(PA-AP)\quad \text{and} \quad D_Q(A)=i(QA-AQ)$$
where $D_P$ and $D_Q$ are unbounded operators on $\mathcal{B}(L_2(\R))$.  Next, define $L:D(L) (\subseteq \mathcal{B}(L_2(\R)) \rightarrow L_2(\R)$ by $L=D_P^2+D_Q^2$.  Then L generates a QMS.
\end{exm}
The fact that L generates a QMS was proved by Arveson in \cite{arveson}.  By expanding L, we have
$$L(A)=2\left(PAP+QAQ \right) - \left( P^2+Q^2 \right)A - A \left( P^2+Q^2 \right)$$
for all $A \in D(L)$.  Note here that this expression is in the form given by Corollary \ref{3.7} with $\ki = \hi \oplus \hi$, $V= \sqrt{2} P \oplus Q$, $\pi (A)= A \oplus A$, and $G=-(P^2+Q^2)$.

Let $e \in L_2( \R)$ of norm one such that $|e \ra \la e| \in D(L)$, say $e(x)=\frac{1}{\sqrt{2 \pi}} \exp{(-x^2/2)}$ for example, let $T=|e \ra \la e|$ and $h=e$ be an associate vector for T.  Since $Th=e$ we have 
$$U= \{ u \in L_2(\R): |u \ra \la e| \in \mathcal{A} \}$$  where $\mathcal{A}$ is the domain algebra of L.  Let
$$U'=\{u \in L_2(\R): u',u'', Qu,Q^2u \in L_2(\R) \} .$$
It is an easy exercise to check that $U' \subseteq U$ and, since the Schwartz class is norm dense in $L_2(\R)$, we have that $U$ is norm dense in $L_2(\R)$. 

\begin{exm}\label{4.2}
{\bf (\cite{parthasarathy}-pg. 258)} Let $(B_t)_{t \geq 0}$ be a standard Brownian motion defined on the probability space $(\Omega , \mathcal{F} , P)$ and define $T_t: \mathcal{B} ( \hi) \rightarrow \mathcal{B} (\hi)$ by
$$T_tA=\mathbb{E} \left[ e^{iB_tV}Ae^{-iB_tV} \right]$$
where $V$ is a self-adjoint operator on $\hi$.  Then $(T_t)_{t \geq 0}$ is a QMS.
\end{exm}
The fact that $T_0=1$ and $T_t(1)=1$ are obvious.  To prove $T_{t+s}=T_tT_s$ start with the identity
$$T_{t+s}A= \mathbb{E} \left[e ^{i(B_{t+s}-B_s)V}e^{iB_sV}Ae^{-iB_sV}e^{-i(B_{t+s}-B_s)V} \right]$$
and use the property of independent increments for Brownian motion to get the desired result.  The remaining properties which qualify $(T_t)_{t \geq 0}$ as a QMS are fairly obvious.
Now, suppose $V$ is bounded.  Let $(t_n)_{n \in \N} \subseteq [0, \infty)$ such that $t_n \rightarrow t$.  Further, for $A \in \mathcal{B} (\hi)$,
\begin{align*}
\| \mathbb{E} \left[ e^{iB_{t_n}V}A e^{-iB_{t_n}V}-e^{iB_tV}Ae^{-iB_tV} \right] \| & \leq \int_{\Omega} \| e^{iB_{t_n}V}A e^{-iB_{t_n}V}-e^{iB_tV}Ae^{-iB_tV} \| dP
\\ & \leq \int_{\Omega} \left( \| e^{i(B_{t_n}-B_t)V}-1 \| + \| e^{i(B_{t_n}-B_t)V}-1 \| \right) dP
\\ & \rightarrow 0
\end{align*}

\noindent by the Bounded Convergence Theorem since $B_{t_n}( \omega) \rightarrow B_t(\omega)$.  So we have that $(T_t)_{t \geq 0}$ is a uniformly continuous QMS.  Next, we claim that $T_tA= \mathbb{E} \left[e^{iB_t(ad V)}A \right]$ where $(ad V)A=VA-AV$ for all $A \in \mathcal{B}(\hi)$.  To this end, it's an exercise to show that
$$ (ad V)^n A= \sum_{k=0}^n (-1)^k \binom{n}{k} V^{n-k}AV^k$$
which gives
$$ \mathbb{E} \left[ e^{iB_t(ad V)}A \right] = T_tA .$$
Further,
$$T_tA=\frac{1}{\sqrt{2} \pi} \sum_{n=0}^{\infty} \int_{\R}e^{-x^2/2} \frac{(ix \sqrt{t})^n}{(2n)!}(ad V)^{2n}A dx .$$
Then, using our knowledge of Gaussian integrals, we'll find that
$$T_tA=e^{-\frac{1}{2}(ad V)^2} A .$$
So the generator L of $(T_t)_{t \geq 0}$ is given by 
$$ L (A) = - \frac{1}{2}(ad V)^2A= \frac{-1}{2} \left( V^2A+AV^2-2VAV \right) .$$

Now, if $V$ is unbounded then the generator is given ``formally'' by the above equation, that is, L can be realized as a sesquilinear form where
$$\la u, L(A)v \ra = \la Vu, AVv \ra + \la u, -\frac{1}{2}V^2Av \ra + \la -\frac{1}{2}V^2A^*u,v \ra .$$

Also, the generator has the form given in Corollary \ref{3.7} with $G=-\frac{1}{2}V^2$.  If $\hi =L_2 ( \R)$ and $V=i \frac{d}{dx}$ then let $e(x)=\frac{1}{\sqrt{2 \pi}} \exp{(-x^2/2)}$ and let $T=|e \ra \la e|$.  Then $h=e$ is an associate vector for T and it is an easy exercise to see that 
$$U= \{ u \in L_2(\R): |u \ra \la e \in \mathcal{A} \} \supseteq \{ f \in L_2(\R) : f',f'' \in L_2(\R) \} ,$$
and therefore U is dense in $L_2( \R)$.

\begin{exm}\label{4.3}
{\bf (\cite{arveson} and similar examples produced in \cite{fagnola2} and \cite{powers})} Let $\hi = L_2[0, \infty)$ and define $U_t: \hi \rightarrow \hi$ by 

\[
(U_tg)(x)=\left\{ \begin{array}{lr}
g(x-t) & \text{if }x \geq 0 \\
0 & \text{otherwise}
\end{array}
\right.
\]

Then $(U_t)_{t \geq0}$ is a strongly continuous semigroup of isometries whose generator $D$ is differentiation.  Let $f \in L_2(0, \infty)$ be what we get by normalizing $u(x)=e^{-x}$ (i.e. $f=\frac{u}{\| u \|})$ then define $\omega : \mathcal{B}(\hi) \rightarrow \C$ by $\omega (A)= \la f,Af \ra$.  Define the completely positive maps $\phi_t: \mathcal{B}(\hi) \rightarrow \mathcal{B}(\hi)$ where
$$\phi_t(A)= \omega(A)E_t+U_tAU_t^*$$ 
for all $t \geq 0$ where $E_t$ is the projection onto the subspace $L_2(0, t) \subseteq L_2(0, \infty)$.  Then $( \phi_t)_{t \geq 0}$ is a QMS.
\end{exm}

First note that for $A \in \mathcal{B}(\hi)$,
$$\omega(U_tAU_t^*)= \la U_t^*f,AU_t^* f \ra = \left\la \frac{e^{-( \cdot +t)}}{\| u \|},A\left( \frac{e^{-( \cdot +t)}}{\| u \|}\right) \right\ra = e^{-2t} \la f,Af \ra = e^{-2t} \omega (A)$$
where the dots denote the variable of the function.
We claim that $(\phi_t)_{t \geq 0}$ is a semigroup.
First, we want to show that $\omega( \phi_t(A)) = \omega(A)$ for all $A \in \mathcal{B}( \hi)$.  Indeed,
\begin{align*}
\omega( \phi_t(A)) & = \omega ( \omega(A)E_t + U_tAU_t^*)
\\ & = \omega (A) \la f,E_t f \ra + \la  f, U_tAU_t^*f \ra 
\\ & = \omega (A) \la f, (1-U_tU_t^*)f \ra +e^{-2t} \omega(A) && \text{since } E_t=1-U_tU_t^*
\\ & = \omega(A)(1-e^{-2t})+e^{-2t} \omega (A) = \omega (A) .
\end{align*}
So we have that $\omega( \phi_t(A)) = \omega(A)$ for all $A \in \mathcal{B}( \hi)$.  Next, we want to show $\phi_s \phi_t = \phi_{t+s}$.  Let $A \in \mathcal{B}(\hi)$.  Then
$$\phi_s \phi_t(A)= \omega( \phi_t(A))E_s+U_s( \omega(A)E_t+U_tAU_t^*)U_s^*=\omega(A)(E_s+U_sE_tU_s^*)+U_{s+t}A(U_{s+t})^*$$
and, since $E_{s+t}=E_s+U_sE_tU_s^*$, we have that
$$\phi_s \phi_t(A)= \omega (A)E_{s+t}+U_{s+t}A(U_{s+t})^*=\phi_{s+t}(A) .$$
So we have that $(\phi_t)_{t \geq 0}$ is a QMS.  If $L$ denotes the generator of $(\phi_t)_{t \geq 0}$ and $D(L)$ denotes the domain of $L$ then 

\begin{align} \label{GA1} 
L(A) & = \sigma-\text{weak}- \lim_{t \to 0}\limits \frac{1}{t} (\phi_t (A) - A) \nonumber \\
& = \sigma-\text{weak}-\lim_{t \to 0}\limits \left( \frac{\omega (A)E_t}{t} + \frac{U_t A U_t^* -A}{t} \right) 
\quad \text{for all }A \in D(L).
\end{align}
By Example~\ref{1.3} the generator of the QMS $(A \mapsto U_t A U_t^*)_{ t \geq 0}$ is equal to $\alpha_D$ where $\alpha_D(A)=DA+AD^*$. 
By \eqref{GA1}, if a bounded operator $A$ belongs to $D(\alpha_D)$ and the kernel of $\omega$, (i.e. $\omega (A)=0$), then 
$A \in D(L)$. Now fix a normalized vector $e \in L_2[0, \infty )$ such that $D(e) \in L_2[0, \infty )$ and $\langle e , f \rangle =0$ 
and use this vector $e$ in place of $Th$ to define the subspace $U$ of Corollary~\ref{3.7}, 
(for example take the positive finite rank operator
$T$ to be equal to $|e\rangle \langle e |$ and the associate vector $h$ of $T$ to be equal to $e$; it is easy to verify that 
$ T \in D(L)$, $ Th =e$, and $|Th\rangle \langle Th | \in D(L)$). Using the fact that 
$\text{ker}\, \omega \cap D(\alpha_D) \subseteq D(L)$,
is  easy to verify that for all 
$x \in L_2[0,\infty)$ with $D(x) \in L_2[0, \infty)$ and $\langle x , f \rangle =0$ we have that the following three conditions 
are satisfied: $|x \rangle \langle e | \in D(L)$, 
$ |x \rangle \langle e | ( |x \rangle \langle e | )^* = \| e \|^2 | x \rangle \langle x | \in D(L)$, and 
$ ( |x \rangle \langle e | )^* |x \rangle \langle e | = \| x \|^2 | e\rangle \langle e | \in D(L)$. Thus by Definition~\ref{domainalgebra} $|x \rangle \langle x | \in \mathcal{A}$ where $\mathcal{A}$ denotes the domain algebra of $L$.
Hence
\begin{equation} \label{GA2}
\{ x \in L_2[0, \infty ): D(x) \in L_2[0, \infty ) \text{ and } \langle x , f \rangle =0 \} \subseteq U.
\end{equation}
Arveson proves \cite[Proposition, pg. 75]{arveson} that the strong operator closure $\overline{\mathcal{A}}^{\text{SOT}}$ of
the domain algebra is equal to the set of bounded operators $A$ such that both $A$ and its adjoint $A^*$ have $f$ as an eigenvector (necessarily corresponding to complex conjugate eigenvalues). Thus for $x \in L_2[0, \infty)$, if 
$A=|x\rangle \langle e| \in \mathcal{A}$ then $\langle f , e \rangle =0$. Therefore
\begin{equation} \label{GA3}
 U \subseteq \{ x \in L_2[0, \infty ) :  \langle x, f\rangle =0 \} .
\end{equation}
We do not have more precise description of $U$ besides \eqref{GA2} and \eqref{GA3}. Equation \eqref{GA3} shows that 
$U$ is not dense in $\mathcal{H}$.
Note that the domain algebra 
$\mathcal{A}$ contains operators which are not in the kernel of $\omega$ 
(since $|f \rangle \langle f | \in \overline{\mathcal{A}}^{\text{SOT}}$ by \cite[Proposition, pg. 75]{arveson}). Hence
the operator $V$ and the unital $^*$-representation $\pi$ which appear in the statement of Corollary~\ref{3.7}
are non-zero.  The operator $G$ which appears in the statement of Corollary~\ref{3.7} is not necessarily equal to the 
generator $D$ of $(U_t)_{t \geq 0}$. Formulas for $V$, $\pi$ and $G$ are given in Corollary~\ref{3.7} and 
Theorems~\ref{3.9} and \ref{3.4} and we do not know simpler formulas for this particular example.

\end{spacing}

%

%
\


\begin{thebibliography}{}
%
%
\bibitem{arveson2}  Arveson, W.:
{Subalgebras of $C^*$-algebras}.
Acta Math. \textbf{123}, 141-224 (1969)

\bibitem{arveson} Arveson, W.:
{The Domain Algebra of a CP-Semigroup}.
Pac. J. Math. \textbf{203(1)} , 67-77 (2002)

\bibitem{arveson3}  Arveson, W.:
{Noncommutative Dynamics and E-Semigroups, Springer Monographs in Mathematics}.
Springer-Verlag, New York (2003)

\bibitem{br} Bratteli, O. and Robinson D.:
{Operator Algebras and Quantum Statistical Mechanics I}.
 Springer-Verlag, New York (1979)

\bibitem{ce} Christensen, E. and Evans, D.E.:
{Cohomology of Operator Algebras and Quantum Dynamical Semigroups}.
J. Lond. Math. Soc. \textbf{20}, 358-368 (1979)


\bibitem{dye-russo} Dye, H. A. and Russo, B.:
{A Note on Unitary Operators in $C^*$-Algebras}.
Duke Math. J. \textbf{33(2)}, 413-416 (1966)

\bibitem{fagnola1} Fagnola, F.:
{Quantum Markov Semigroups and Quantum Flows}.
Proyecciones \textbf{18(3)},1-144 (1999)

\bibitem{fagnola2} Fagnola, F.:
{A Simple Singular Quantum Markov Semigroup}.
In: Rebolledo, R. (eds.)
{Stochastic Analysis and Mathematical Physics}, pp. 73-87.
Birkhauser, Boston (2000)

\bibitem{gks} Gorini, V., Kossakowski, A and Sudarshan, E.C.G.
{Completely Positive Dynamical Semigroups of N-Level Systems}.
J. Math. Phys. \textbf{17}, 821 (1976)

\bibitem{hille-phillips} Hille, E. and Phillips, R.S.:
{Functional Analysis and Semigroups}.
Amer. Math. Soc., Providence (1957)

\bibitem{ingarden-kossakowski} Ingarden, R. S. and Kossakowski, A.:
{On the Connection of Nonequilibrium Information Thermodynamics with Non-Hamiltonian Quantum Mechanics of Open Systems}.
Ann. Phys. \textbf{89},451-485 (1975)

\bibitem{kossakowski}  Kossakowski, A.:
{On Quantum Statistical Mechanics of Non-Hamiltonian Systems}.
Rep. Math. Phys. \textbf{3}, 247-274 (1972)

\bibitem{kraus} Kraus, K.:
{General State Changes in Quantum Theory}.
Ann. Phys. \textbf{64},311-335 (1970)

\bibitem{lindblad} Lindblad, G.:
{On the Generators of Quantum Dynamical Semigroups}.
Commun. math. Phys. \textbf{48},119-130 (1976)

\bibitem{parthasarathy} Parthasarathy, K. R.:
{An Introduction to Quantum Stochastic Calculus}.
 Birkhauser, Basel (1992)

\bibitem{pedersen} Pedersen,G. K.:
{Analysis Now}.
 Springer-Verlag, New York (1989)

\bibitem{powers2}  Powers, R. T.:
{A Non-spacial Continuous Semigroup of $^*$-endomorphisms of $\mathcal{B}(\hi)$}.
Publ. RIMS (Kyoto University) \textbf{23(6)}, 1054-1069 (1987)

\bibitem{powers} Powers, R. T.:
{New Examples of Continuous Spatial Semigroups of Endomorphisms of $\mathcal{B}(\hi)$}.
Int. J. Math. \textbf{10},215-288 (1999)

\bibitem{stinespring} Stinespring, W. F.:
{Positive Functions on $C^*$-algebras}.
Proc. Amer. Math. Soc., 211-216 (1955)

\bibitem{sunder} Sunder, V. S.:
{Functional Analysis: Spectral Theory}.
Birkhauser, Berlin (1998)

\bibitem{topping}  Topping, D.M.:
{Lectures on von Neumann Algebras}.
Van Nostrand, London (1971)

\end{thebibliography}
\end{document}